\documentclass[letterpaper, 10 pt, conference]{ieeeconf}  
\usepackage{amsmath}
\usepackage{graphicx, subfigure}

\IEEEoverridecommandlockouts                              

\newtheorem{prop}{Proposition}

\title{\LARGE \bf
Hopf Bifurcations in Replicator Dynamics with Distributed Delays
}

\author{Nesrine Ben-Khalifa$^{\diamond}$, Rachid El-Azouzi$^{\diamond}$ and Yezekael Hayel$^{\diamond}$ 
\thanks{$^\diamond$ CERI/LIA, University of Avignon, France.}
}

\def\be{\begin{equation}}
\def\ee{\end{equation}}

\def\bearn{\begin{eqnarray*}}
\def\eearn{\end{eqnarray*}}
\def\bear{\begin{eqnarray}}
\def\eear{\end{eqnarray}}
\def\barr{\begin{array}}
\def\earr{\end{array}}
\usepackage{subfigure}%
\usepackage{epstopdf}
\begin{document}

\maketitle
\thispagestyle{empty}
\pagestyle{empty}

\begin{abstract}
In this paper, we study the existence and the property of the Hopf bifurcation in the two-strategy replicator dynamics with distributed delays. In evolutionary games, we assume that a strategy would take an uncertain time delay to have a consequence on the fitness (or utility) of the players. As the mean delay increases, a change in the stability of the equilibrium (Hopf bifurcation) may occur at which a periodic oscillation appears. We consider Dirac, uniform, Gamma, and discrete delay distributions, and we use the Poincar\'e- Lindstedt's perturbation method to analyze the Hopf bifurcation. Our theoretical results are corroborated with numerical simulations.
\end{abstract}

\section{INTRODUCTION}

Evolutionary game theory provides an analytical framework for modeling and studying continuous interactions in a large population of agents. 
It was introduced in \cite{smith_orig} where the authors proposed the notion of an evolutionarily stable strategy (ESS), which is an equilibrium characterized by the property of resistance to a sufficiently small fraction of mutants. However, this static notion does not give an estimate of the time required for the ESS to overcome the mutants, neither the asymptotic state of the population (asymptotic fraction of each strategy in the population) given an initial configuration. To overcome these shortcomings, the authors proposed in \cite{taylorjonker78} the replicator dynamics which is a model that enables the prediction of the time evolution of each strategy's fraction in the population. 
Its main concept is the replicator dynamics \cite{taylorjonker78} which is a model that enables the prediction of the time evolution of each strategy's fraction in the population. 
In this dynamics, the growth rate of a strategy is proportional to how well this strategy performs relative to the average payoff in the population.
In social sciences, the replicator dynamics can be seen as an imitation process, where each player has occasionally the opportunity to revise his strategy and imitate another player whose utility is better than his. 

As originally defined, the replicator dynamics does not take into account of delays, and it assumes that an interaction has an instantaneous effect on the fitness of players. However, this assumption can fail to be true. In many situations, we observe some time interval between the use of a strategy and the time a player feels its impact. For example, in economics, the investments take some time delay, which is usually uncertain, to generate revenues. In evolutionary game literature, there have been works that included delays in the replicator dynamics such as \cite{yi97,hopf_2012,hopf2016,cdc2014,albos,delay2012}. In \cite{yi97}, the author introduced a single fixed delay in the fitness function and derived the critical delay at which the stability of the equilibrium is lost. In \cite{nesrine_random}, a linear analysis of the replicator dynamics with distributed delays is proposed. In \cite{hopf2016}, the authors examined the Hopf bifurcations in the replicator dynamics considering a scenario where the delay is fixed but not all the interactions are equally subject to delays. They distinguished homogeneous interactions (involving similar strategies) and mixed interactions (involving different strategies). 

A Hopf bifurcation can be supercritical, in which case  the limit cycle is stable, or subcritical, in which case the limit cycle is unstable. In this work, we aim to determine the properties of  the limit cycle created in the neighborhood of the Hopf bifurcation using Poinacar\'e-Lindstedt's perturbation method. To the best of our knowledge, this paper is the first attempt to study the bifurcations in the replicator dynamics with distributed delays.

The present paper is structured as follows. First, in Section \ref{sec0}, we recall the main concepts in evolutionary games. In Section \ref{sec_repl}, we derive the replicator dynamics with distributed delays. In Section \ref{hopf}, we analyze the Hopf bifurcations in the replicator dynamics considering Dirac, uniform, Gamma, and discrete delay distributions. In Section \ref{sec_num}, we compare the theoretical results with numerical simulations. Finally, in Section \ref{sec_conc}, we conclude the paper.

\section{Evolutionary Games}\label{sec0}
We consider a population in which the agents are continuously involved in random pairwise interactions. At each interactions, the engaged players obtain payoffs that depend on the strategies used. The matrix that gives the outcome of an interaction for both players is given by:
\begin{equation*}
{\cal{G}}=\left (
\begin{array}{cc}
a&b\\
c&d
\end{array}
\right).
\end{equation*}
Let $s(t)$ ($1-s(t)$) be the proportion of the population using the strategy $A$ ($B$). The utilities of strategies $A$ and $B$ at an instant $t$ are given by:
\begin{eqnarray*}
U_A(t)&=&as(t)+b(1-s(t)),\\
U_B(t)&=&cs(t)+d(1-s(t)).
\end{eqnarray*}
Let $\bar{U}$ be the average payoff in the population. $\bar{U}(t)$ is given by:
\begin{eqnarray*}
\bar{U}(t)=s(t)U_A(t)+(1-s(t))U_B(t).
\end{eqnarray*}
In two-strategy games, the replicator dynamics is given by:
\begin{eqnarray*}
\frac{ds(t)}{dt}&=&s(t)(U_A(t)-\bar{U}(t)),\\
&=&s(t)(1-s(t))(U_A(t)-U_B(t)).
\end{eqnarray*}
Let $\delta_1=b-d$, $\delta_2=c-a$, and $\delta=\delta_1+\delta_2$. If $\delta_1>0$ and $\delta_2>0$, then there exists a mixed equilibrium given by $s^*=\frac{\delta_1}{\delta_1+\delta_2}$ which is asymptotically stable in the replicator dynamics. We consider hereafter that this assumption holds.
\section{Replicator dynamics with Distributed Delays}\label{sec_repl}
In this section, we introduce in the replicator dynamics continuous distributed delays. When a player uses a strategy at time $t$, he would receive his payoff after some random delay $\tau$, it means at time $t+\tau$. Then its expected utility is determined only at that instant, i.e. $U(t+\tau)$. 
If the delay is equal to $\tau$, then the expected payoff of strategy $A$ at time $t$ is determined by:
$$
U_A(t,\tau)=as(t-\tau)+b(1-s(t-\tau)),
$$
if $t\geq \tau$, it is 0 otherwise.
Let $p(\tau)$ be the probability distribution of delays whose support is $[0,\infty[$. As we consider a large population, every player can experience a different positive delay. Thus, we consider the expected payoff of all the players choosing strategy $A$ by averaging the payoffs of all individuals and then all possible delays as:
$$
U_A(t)=\int_{0}^{\infty}{p(\tau)}U_A(t,\tau)d\tau.
$$  
The expected payoffs to strategies $A$ and $B$ then write:
\begin{eqnarray*}
U_A(t)=a \int_{0}^{\infty}{p(\tau)s(t-\tau)d\tau}+b\big(1-\int_{0}^{\infty}{p(\tau)s(t-\tau)d\tau}\big),\\
U_B(t)=c \int_{0}^{\infty}{p(\tau)s(t-\tau)d\tau}+d\big(1-\int_{0}^{\infty}{p(\tau)s(t-\tau)d\tau}\big).
\end{eqnarray*}
Therefore, the replicator dynamics can be written as:
\begin{eqnarray}\label{full_equation0}
\frac{ds(t)}{dt}=s(t)(1-s(t))(-\delta \int_{0}^{{\infty}}{p(\tau)s(t-\tau) d\tau}+\delta_1).
\end{eqnarray}
In order to investigate the local asymptotic stability of the mixed equilibrium $s^*$, we can make a linearization around the equilibrium and derive the associated characteristic equation. The equilibrium point of the linearized equation is locally asymptotically stable if and only if all the roots of the characteristic equation have negative real parts; if there exists a root with a zero or positive real part, then it is not asymptotically stable \cite{gopalsamy}.

Let $x(t)=s(t)-s^*$. Substituting $s$ with $x$ in the previous equation, we get:
\begin{eqnarray}\label{full_equation}
\frac{dx(t)}{dt}=-\delta \gamma  \int_{0}^{\infty}{p(\tau)x(t-\tau) d\tau}-\delta (1-2s^*) x(t) \times  \nonumber \\ \int_{0}^{\infty}{p(\tau)x(t-\tau) d\tau}+  \delta x^2(t)  \int_{0}^{\infty}{p(\tau)x(t-\tau) d\tau};
\end{eqnarray}
which is of the form,
\begin{eqnarray}\label{full_eq2}
\frac{dx(t)}{dt}=A\int_{0}^{\infty}{p(\tau)x(t-\tau) d\tau}+B x(t) \times \nonumber \\ \int_{0}^{\infty}{p(\tau)x(t-\tau) d\tau}+   C  x^2(t)  \int_{0}^{\infty}{p(\tau)x(t-\tau) d\tau},
\end{eqnarray}
where $A= -\delta \gamma$, $B=-\delta (1-2s^*)$, $C=\delta$, and $\gamma=s^*(1-s^*)$. 
Keeping only linear terms in the above equation, we get the following linearized equation:
\begin{eqnarray}
\frac{dx(t)}{dt}&=&A\int_{0}^{\infty}{p(\tau)x(t-\tau) d\tau}.
\end{eqnarray}
The characteristic equation can be derived by taking the Laplace transform of the linearized equation and is given by:
\begin{eqnarray}
\lambda-A\int_{0}^{\infty}{p(\tau)exp(\lambda \tau) d\tau}=0.
\end{eqnarray}
The characteristic equation enables us to determine the local asymptotic stability of the equilibrium. When a pair of conjugate complex roots passes through the imaginary axis, a Hopf bifurcation occurs at which the asymptotic stability of the equilibrium is lost and a limit cycle is created. The frequency of oscillations of the bifurcating limit cycle is equal to the complex parts of the pure imaginary root \cite{gopalsamy}. The amplitude of the limit cycle is very small at the Hopf bifurcation and grows gradually as the mean delay increases.

To determine the properties of the bifurcating limit cycle (criticality and amplitude), we should take into account of all the terms in the replicator dynamics, including nonlinear terms and a perturbation method should be used. We propose to use the Lindstedt's method, which has been proved to be efficient \cite{hopf2007,hopf2016}. The Lindstedt's method enables us to have an approximation of the bifurcating limit cycle.

At the Hopf bifurcation, the solution of the replicator dynamics (\ref{full_equation}) can be approximated as \cite{hopf2007}:
\begin{eqnarray*}
x(t)=A_m\mbox{cos}(w_0 t).
\end{eqnarray*}
To examine the bifurcating solution, we define a small parameter $\epsilon$ and a new variable $u$ as follows:
\begin{eqnarray*}
x(t)=\epsilon u(t).
\end{eqnarray*}
Furthermore, we stretch time by defining a new variable $\Omega$ as follows:
\begin{eqnarray*}
T=\Omega t.
\end{eqnarray*}
The equivalent replicator dynamics can then be written as:
\begin{eqnarray}\label{eq_tot}
&& \hskip -0.7 cm \Omega \frac{du(T)}{dT}=A \int_{0}^{\infty}{p(\tau)u(T-\Omega \tau) d\tau}+\epsilon B u(T)  \int_{0}^{\infty}{p(\tau)}\times \nonumber \\ &&\hskip -0.5cm {u(T-\Omega \tau) d\tau}+ \epsilon^2 C  u^2(T)\int_{0}^{\infty}{p(\tau)u(T-\Omega \tau) d\tau}. 
\end{eqnarray}
In addition, we make a series expansion of $\Omega$ as follows:
\begin{eqnarray}
\Omega=w_0+ \epsilon^2k_2+{{O}}(\epsilon^3),
\end{eqnarray}
we remove the ${O}(\epsilon)$ term because it turns out to be a zero, and
\begin{eqnarray}
u(T)=u_0(T)+\epsilon u_1(T)+\epsilon^2 u_2(T)+{{O}}(\epsilon^3).
\end{eqnarray}
Finally, in equation (\ref{eq_tot}), we expand out and compare the terms of the same order in  $\epsilon$. By setting the secular terms to zero, we get the amplitude of the limit cycle.
\section{Hopf Bifurcations in the Replicator Dynamics}\label{hopf}
We aim to determine the properties of the Hopf bifurcation in the replicator dynamics subjet to distributed delays. In the next subsections, we consider different delay distributions: Dirac, uniform, and Gamma distributions. We also consider a case of stochastic and discrete delays.
\subsection{Dirac Distribution}
We suppose there is a single fixed delay of a value $\tau$. The replicator dynamics in (\ref{full_equation0}) reduces to:
\begin{eqnarray*}
\frac{ds(t)}{dt}=s(t)(1-s(t))\big(-\delta s(t-\tau)+\delta_1\big).
\end{eqnarray*}
The characteristic equation associated to the linearized replicator dynamics around the interior equilibrium is given by:
\begin{eqnarray*}
\lambda+\delta \gamma exp(-\lambda \tau)=0.
\end{eqnarray*}
Let $\lambda^*=iw_0$ be a pure imaginary root, from the characteristic equation, that is:
\begin{eqnarray*}
iw_0+\delta \gamma exp(-iw_0\tau_{cr})=0,
\end{eqnarray*}
By separating the real and imaginary parts, we obtain:
\begin{eqnarray*}
\tau_{cr}=\frac{\pi}{2 \delta \gamma}, \mbox{ and } w_0=\delta \gamma.
\end{eqnarray*}
These formulae will be used later in solving the DDE.
By making a change of variable as mentioned in the previous section, we can write the replicator dynamics (\ref{eq_tot}) as follows:
\begin{eqnarray}\label{repl_dirac1}
\Omega \frac{du(T)}{dT}&=&Au(T-\Omega \tau)+\epsilon B u(T-\Omega \tau)+ \epsilon^2 C u^2(T)\times    \nonumber \\ && u(T-\Omega \tau).
\end{eqnarray}
From the replicator dynamics (\ref{repl_dirac1}), we can examine the behavior of the bifurcating periodic solution. The following proposition summarizes the properties of the bifurcating limit cycle.
\begin{prop}\label{prop_dirac}
Let $P=-20A^3 $ and $Q=5A^2C\tau_{cr}-3AB^2 \tau_{cr}-B^2$, the amplitude of the bifurcating limit cycle is given by
\begin{eqnarray*}
A_m=\sqrt{\frac{P}{Q}\mu},
\end{eqnarray*}
where $\mu=\tau-\tau_{cr}$. Furthermore, the Hopf bifurcation is supercritical.
\end{prop}
\begin{proof}
See the Appendix.
\end{proof}
Our result is in coherence with the results in \cite{hopf2007}. The result above means that the amplitude of the bifurcating limit cycle is proportional to $\sqrt{\tau-\tau_{cr}}$.
\subsection{Uniform Distribution}
When the delays are i.i.d. random variables drawn from the uniform distribution, that is when,
 \begin{eqnarray*}
p(\tau)=\frac{1}{\tau_{max}} {\mbox{ for }} \tau \in [0,\tau_{max}], {\mbox{ and }} {\mbox {zero otherwise} },
\end{eqnarray*}
the replicator dynamics can be written as:
\begin{eqnarray*}
\frac{ds(t)}{dt}=s(t)(1-s(t))\big(-\delta \int_{0}^{\tau_{max}}{\frac{1}{\tau_{max}}s(t-\tau) d\tau}+b-d\big).
\end{eqnarray*}
The corresponding linearized equation is given by:
\begin{eqnarray*}
\frac{dx(t)}{dt}= \frac{A}{\tau_{max}} \int_{0}^{\tau_{max}}{x(t-\tau) d\tau},
\end{eqnarray*}
and the associated characteristic equation is given by:
\begin{eqnarray*}
\lambda -\frac{A}{\tau_{max}}  \int_{0}^{\tau_{max}}{exp(-\lambda \tau) d\tau} =0.
\end{eqnarray*}
At the Hopf bifurcation, we have $\tau_{cr}=\frac{\pi^2} {2D}$ and $w_0=\frac{\pi}{\tau_{cr}}$,
where $D=\gamma \delta$.
In this case, equation (\ref{eq_tot}) is determined by:
\begin{eqnarray}\label{repl_unif}
&& \hskip -0.7cm \Omega \frac{du(T)}{dT}=\frac{A}{\tau_{max}} \int_{0}^{\tau_{max}}{  u(T-\Omega \tau) d\tau}+ \frac{\epsilon B}{\tau_{max}} u(T)\times\nonumber \\  && \hskip -1.01cm \int_{0}^{\tau_{max}}{ }{  u(T-\Omega \tau) d\tau}+ \frac{\epsilon^2 C}{\tau_{max}} u^2(T) \int_{0}^{\tau_{max}}{ u(T- \Omega \tau) d\tau}.\;  \;\;\;\;\;\;
\end{eqnarray}
From this equation, we can determine the properties of the bifurcating limit cycle, which are brought out in the next proposition.
\begin{prop}\label{prop_unif}
Let $P=8A^2$ and $Q=\tau_{cr}(B^2-2AC)$. The amplitude of the bifurcating limit cycle is given by:
\begin{eqnarray}
{A}_m=\sqrt{\frac{P}{Q}\mu},
\end{eqnarray}
where $\mu=\tau_{max}-\tau_{cr}$. Furthermore, the Hopf bifurcation is supercritical.
\end{prop}
\begin{proof}
See the Appendix.
\end{proof}
\begin{figure}
\hskip -0.565cm \includegraphics[width=5cm]{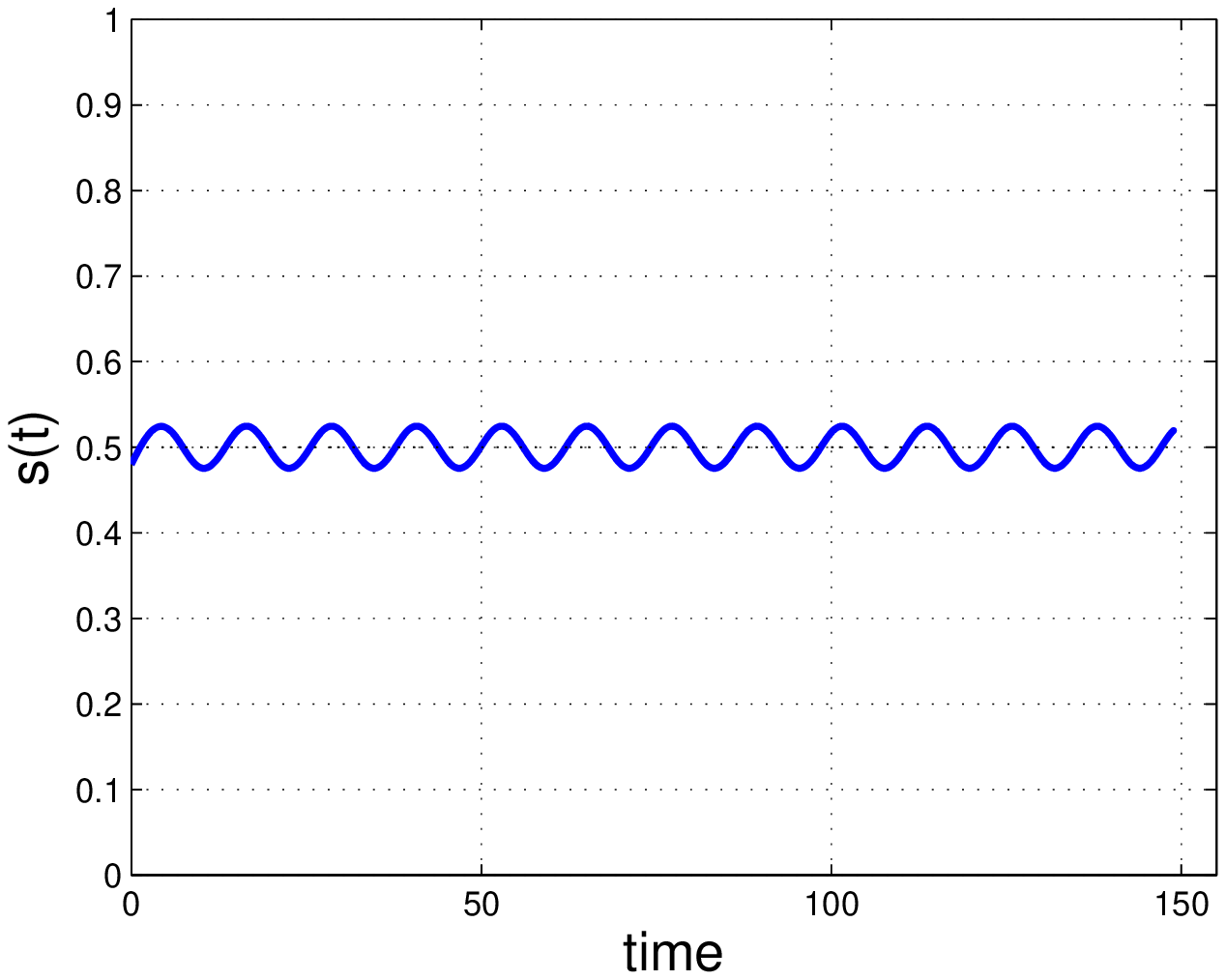}\includegraphics[width=5cm]{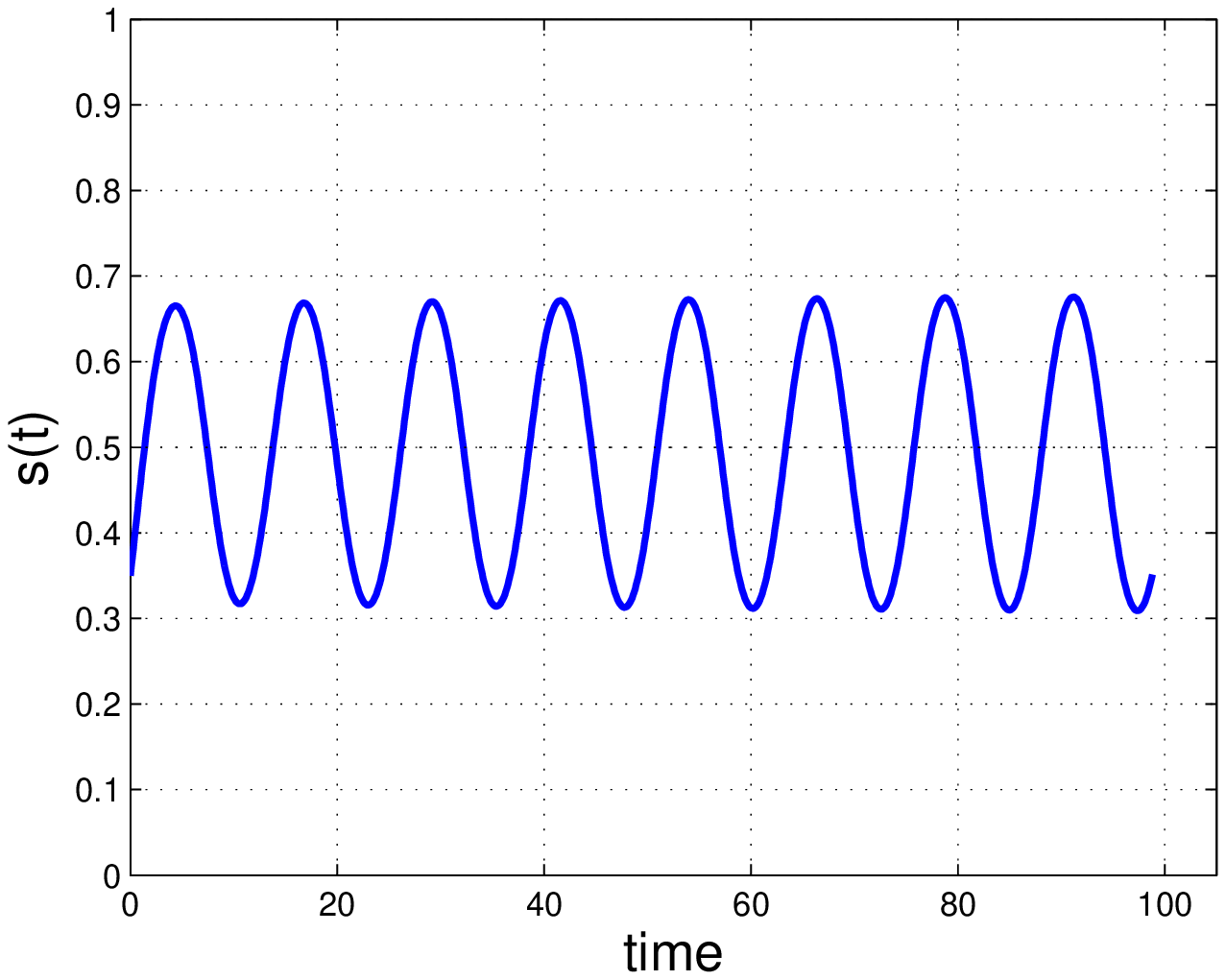} \caption{Stable limit cycle with the uniform distribution with $\tau_{cr}=6.58$ time units. {\it{Left}} $\mu=0.001$, $\tau_{max}=\tau_{cr}+\mu$. {\it{Right}} $\mu=0.03$, $\tau_{max}=\tau_{cr}+\mu$, where $a=-0.5$, $b=3$, $c=0$, and $d=1.5$.}\label{fig_limit_cycle}
\end{figure}
The amplitude of the bifurcating periodic solution is proportional to $\sqrt{\tau_{max}-\tau_{cr}}$. When the value of $\tau_{max}$ is near and superior to $\tau_{cr}$, the replicator dynamics exhibits a stable periodic oscillation in the proportions of the strategies in the population.
We illustrate in Fig. \ref{fig_limit_cycle}, the stable limit cycle occurring near the Hopf bifurcation under the uniform distribution. In the left-subfigure, we fixed $\mu$ to $0.001$ time units whereas in the right-subfigure, $\mu$ is fixed to $0.03$ time units. We recall that $\mu=\tau_{max}-\tau_{cr}$. In the first case, we observe that the stable limit cycle has a very small amplitude, and by increasing $\tau_{max}$, we see in the second case a limit cycle with an amplitude of approximately $0.18$. The amplitude of the oscillation, indeed, increases significantly as $\tau_{max}$ moves away from $\tau_{cr}$.  
 \subsection{Gamma Distribution}
We consider a Gamma distribution of delays with support $[0, \infty[$ and parameters $k\geq1$ and $\beta>0$. The probability distribution in this case is given by:
\begin{eqnarray*}
p(\tau; k,\beta)=\frac{\beta^k \tau^{k-1} e^{-\beta \tau}}{\Gamma(k)},
\end{eqnarray*}
where $\Gamma(k)=(k-1)!$ (Gamma function). The mean of the Gamma distribution is $\frac{k}{\beta}$.
%

The characteristic equation associated to the linearized replicator dynamics is given by:
\begin{eqnarray*}
\lambda+D \int_{0}^{\infty}{\frac{\beta^k}{\Gamma(k)}\tau^{k-1}e^{-(\beta+\lambda)\tau}d\tau}=0,
\end{eqnarray*}
where $D=\delta \gamma$. We take as a bifurcation parameter $\beta$. First, let us determine the critical value of this parameter, $\beta_c$, at which the asymptotic stability of the mixed equilibrium is lost. 
A Hopf bifurcation occurs when $\lambda=iw_0$ with $w_0>0$, passes through the imaginary axis, that is when,
\begin{eqnarray*}
iw_0+D\int_{0}^{\infty}{\frac{\beta_c^k}{\Gamma(k)}\tau^{k-1}e^{-(\beta_c+iw_0)\tau}d\tau}=0.
\end{eqnarray*}
Or equivalently when,
\begin{eqnarray*}
iw_0+\frac{D \beta_c^k}{\Gamma(k)}\int_{0}^{\infty}{\tau^{k-1}e^{-(\beta_c+iw_0)\tau}d\tau}=0.
\end{eqnarray*}
By defining a new variable $z=(\beta_c+iw_0)\tau$, we can write the previous equation as:
\begin{eqnarray}\label{eq_char_00}
iw_0+ \frac{D\beta_c^k}{(\beta_c+iw_0)^k}=0.
\end{eqnarray}
Or equivalently, by using the polar form,
\begin{eqnarray}\label{t_eq}
iw_0+D\beta_c^k(\beta_c^2+w_0^2)^{-\frac{k}{2}}exp(-ik\theta)=0,
\end{eqnarray}
where $\theta  \in [0,\frac{\pi}{2}]$, $\mbox{cos}(\theta)=\frac{\beta_c}{(\beta_c^2+w_0^2)^{\frac{1}{2}}}$, and $\mbox{sin}(\theta)=\frac{w_0}{(\beta_c^2+w_0^2)^{\frac{1}{2}}}$.
Separating the real and imaginary parts in (\ref{t_eq}), we derive,
\begin{eqnarray*}
\mbox{cos}(k\theta)&=&0,\nonumber\\
w_0-D\beta_c^k (\beta_c^2+w_0^2)^{-\frac{k}{2}}\mbox{sin}(k\theta)&=&0.
\end{eqnarray*}
Taking account of the previous equations, we finally get:
\begin{eqnarray}\label{mean_cri}
\beta_c=D\displaystyle\frac{{\rm cos}^{k+1}(\frac{\pi}{2k})}{\mbox{sin}(\frac{\pi}{2k})}.  
\end{eqnarray}
The frequency of oscillations of the bifurcating solution is given by:
\begin{eqnarray}\label{freq2}
w_0=D\mbox{cos}^k(\frac{\pi}{2k}).
\end{eqnarray}
As a remark, we observe that when $k=1$ the Gamma distribution coincides with the exponential distribution, and $\beta_c=0$. Therefore, there cannot exist a Hopf bifurcation in this case. In the following, we suppose that $k\geq 2$.
Furthermore, we derive from equation (\ref{eq_char_00}) by implicit differentiation,
\begin{eqnarray*}
{\cal{R}}e\frac{d \lambda(\beta)}{d\beta}_{|{\beta=\beta_c}}<0.
\end{eqnarray*}
Therefore, when $\beta=\beta_c$ a Hopf bifurcation occurs. 
\begin{figure}[t]
\hskip -0.5cm \includegraphics[width=4.8cm]{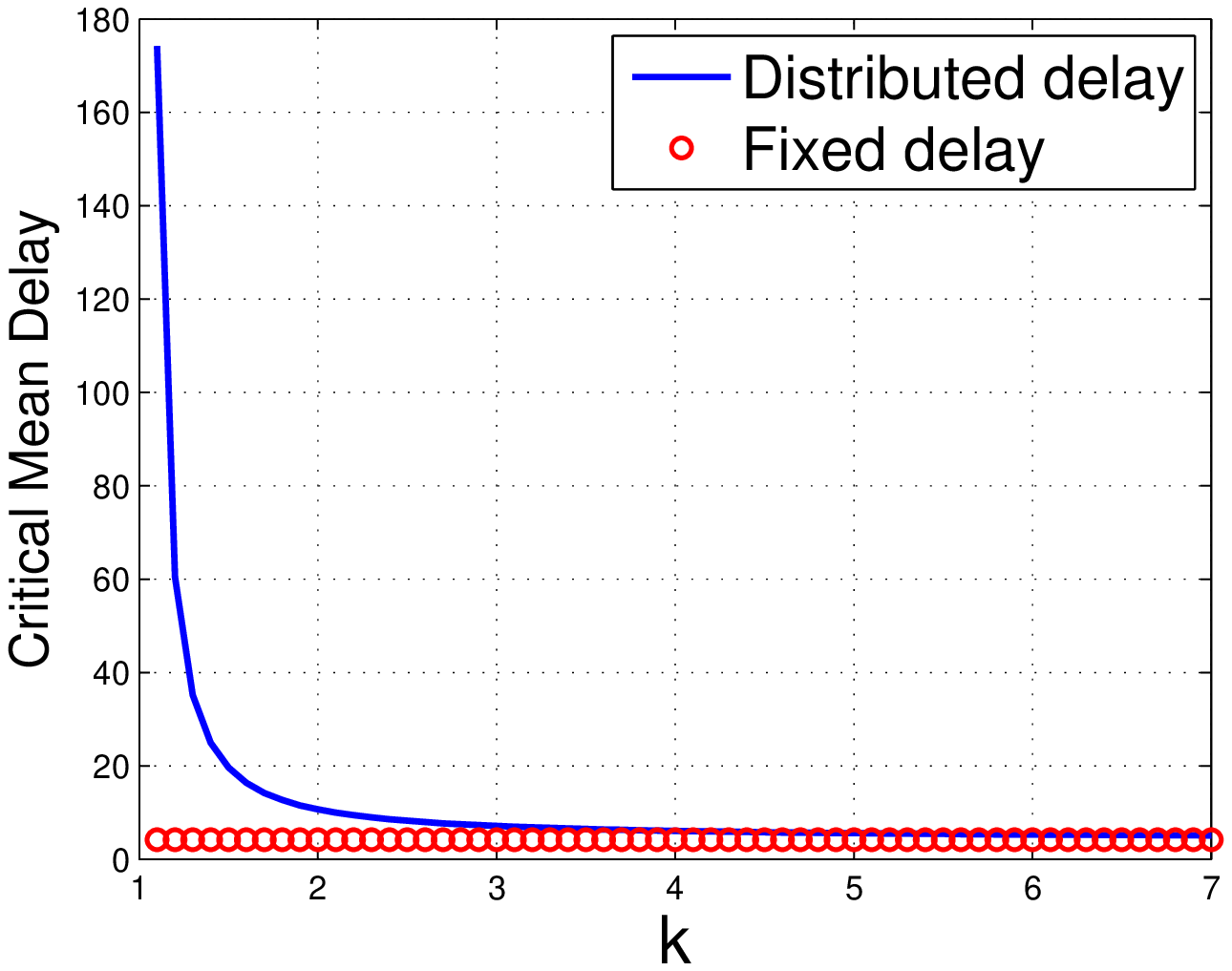}\includegraphics[width=4.8cm]{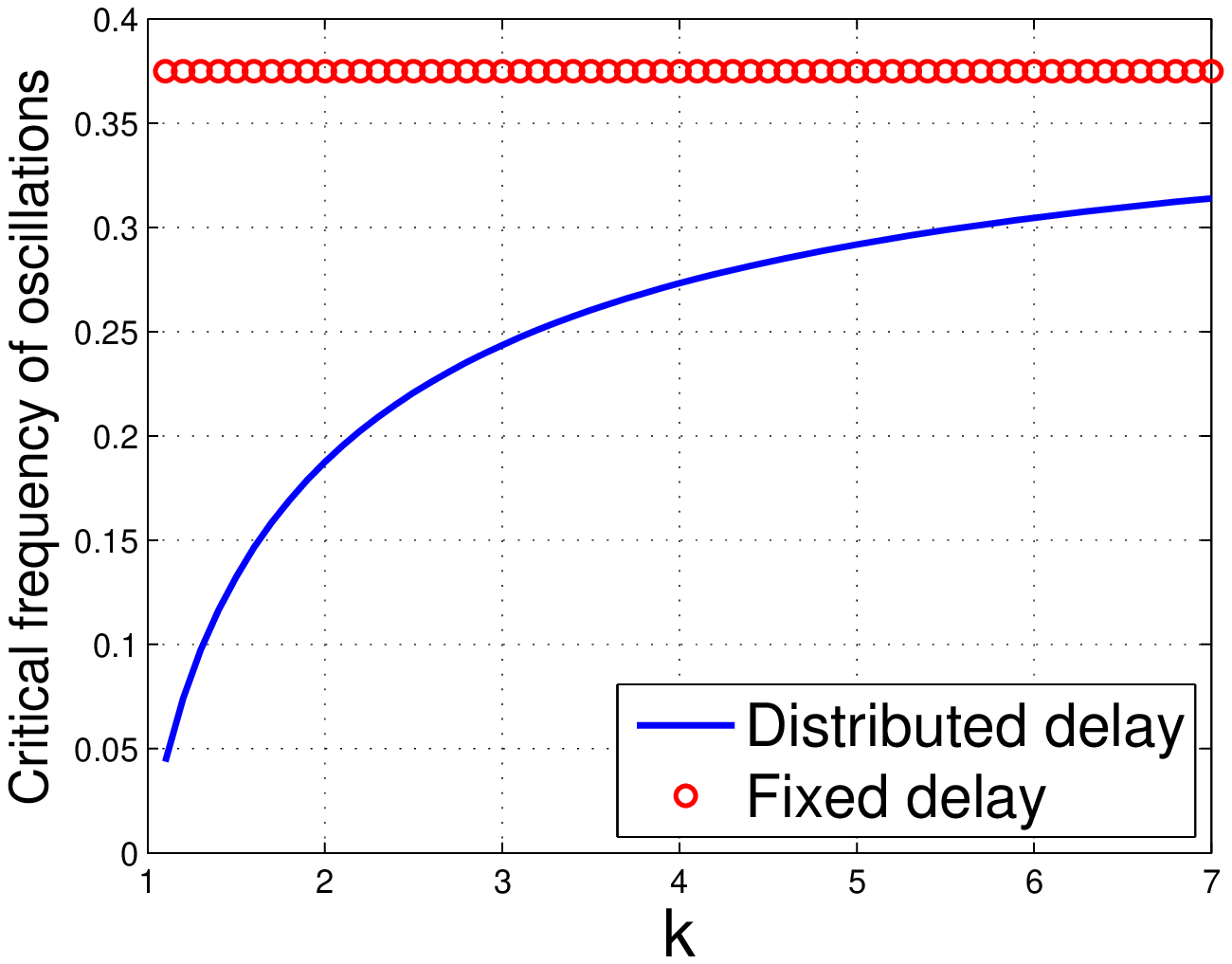}
\caption{{\it{Left}}, the critical mean delay in function of the parameter $k$ under the Gamma distribution. {\it{Right}}, the critical frequency of oscillations in function of the parameter $k$ under the Gamma distribution, where $a=-0.5$, $b=3$, $c=0$, and $d=1.5$.}\label{gamma_mean}
\end{figure}
In Fig. \ref{gamma_mean}-{\it{left}}, we display the critical mean delay $\frac{k}{{\beta_c}}$ with $\beta_c$ given in (\ref{mean_cri}) in function of the parameter $k$. The critical mean delay decreases significantly as the parameter $k$ increases, that is the instability becomes more probable as $k$ grows.

Now, let us determine the properties of the limit cycle in the neighborhood of the bifurcation. We  define ${{I}}$ as follows:
\begin{eqnarray}\label{eq_0001}
{{I}}=\int_{0}^{\infty}{\frac{{\beta}^k}{\Gamma(k)} \tau^{k-1}e^{-\beta \tau}u(T-\Omega \tau) d\tau}.
\end{eqnarray}
The equation (\ref{eq_tot}) can then be written as,
\begin{eqnarray}
\Omega \frac{du(T)}{dT}&=&A {{I}}+\epsilon B u(T) {{I}} + \epsilon^2 C  u^2(T){{I}},
\end{eqnarray}
where $\Omega=w_0+k_2 \epsilon^2+{{O}}(\epsilon^3)$. 
The properties of the bifurcating limit cycles are given in the following proposition.
\begin{prop}\label{prop_gamma}
We consider the following parameters,
\begin{eqnarray*}
P= (k+1) \frac{A}{\beta_c}(1+\frac{w_0^2}{\beta_c^2})^{-\frac{k}{2}}-\frac{k-1}{k+1}(1+\frac{w_0^2}{\beta_c^2})^{\frac{1}{2}}-\frac{w_0}{\beta_c},
\end{eqnarray*}
\begin{eqnarray*}
Q=\frac{B\beta_c}{2(k+1)A} (1+\frac{w_0^2}{\beta_c^2})^{\frac{1}{2}} (F_1 \frac{w_0}{\beta_c}+F_2)-B \frac{w_0}{\beta_c} \times \\ (1+\frac{w_0^2}{\beta_c^2})^{-\frac{k+1}{2}}(F_2+\frac{F_1}{2}(\frac{w_0}{\beta_c}-1))+\frac{C}{4}(1+\frac{w_0^2}{\beta_c^2})^{-\frac {k}{2}},
\end{eqnarray*}
\begin{eqnarray*}\label{f_1}
\hskip -0.2cm{{F}}_1=-\frac{\frac{AB}{2} (1+\frac{w_0^2}{\beta_c^2})^{-\frac{k}{2}} (1+4\frac{w_0^2}{\beta_c^2})^{-\frac{k}{2}} \mbox{cos}(k\theta_1)}{4w_0^2+A^2(1+4\frac{w_0^2}{\beta_c^2})^{-k}+4w_0A (1+4\frac{w_0^2}{\beta_c^2})^{-\frac{k}{2}}\mbox{sin}(k\theta_1)},
\end{eqnarray*}
and,
\begin{eqnarray*}\label{f_2}
\hskip -0.2cm {{F}}_2=-\frac{\frac{B}{2} (1+\frac{w_0^2}{\beta_c^2})^{-\frac{k}{2}}(2w_0+A (1+4\frac{w_0^2}{\beta_c^2})^{-\frac{k}{2}} \mbox{sin}(k\theta_1)) }{4w_0^2+A^2(1+4\frac{w_0^2}{\beta_c^2})^{-k}+4w_0A (1+4\frac{w_0^2}{\beta_c^2})^{-\frac{k}{2}}\mbox{sin}(k\theta_1)}.
\end{eqnarray*}
The amplitude of the bifurcating limit cycle is given by:
\begin {eqnarray}\label{amp_gamma}
A_m=\sqrt{\frac{P}{Q} \mu},
\end{eqnarray}
where $\mu=\beta-\beta_c$ and $\theta_1=\mbox{atan}(\frac{2w_0}{\beta_c})$. Furthermore, the Hopf bifurcation is supercritical.
\end{prop}
\begin{proof}
See the Appendix.
\end{proof}
As in the previous sections, the amplitude of the bifurcating limit cycle is proportional to $\sqrt{\beta_c-\beta}$. Note that the bifurcation occurs when $\beta$ is near $\beta_c$ and $\beta<\beta_c$, therefore the quotient $\frac{P}{Q}$ should be negative. When $\beta$ is near and below $\beta_c$, the replicator dynamics exhibits a stable periodic oscillation in the proportion of strategies in the population.
\subsection{Discrete Delays}
We suppose in this section that a strategy, either $A$ or $B$, would take a delay $\tau$ with probability $p$ or no delay with probability $1-p$. In this case, the replicator dynamics is given by:
\begin{eqnarray*}\label{eq_8899}
\hskip -0.1cm \frac{ds(t)}{dt}=s(t)(1-s(t))\big(- p \delta s(t-\tau)-(1-p) \delta  s(t)+\delta_1\big).
\end{eqnarray*}
Let $x(t)=s(t)-s^*$. Substituting $s$ with $x$ in the replicator dynamics, we get:
\begin{eqnarray*}
\frac{dx(t)}{dt}=-(1-p)\delta \gamma x(t) -p\delta \gamma x(t-\tau)-p\delta (1-2s^*)\times \nonumber \\ x(t) x(t-\tau)- (1-p)\delta(1-2s^*)x^2(t)+p \delta x(t-\tau) x^2(t)\nonumber \\ +(1-p)\delta x^3(t),
\end{eqnarray*}
which is of the form,
\begin{eqnarray*}
\frac{dx(t)}{dt}=a_1 x(t)+b_1 x(t-\tau)+c_1 x(t) x(t-\tau)+d_1x^2(t)+\nonumber \\e_1 x(t-\tau) x^2(t)+f_1 x^3(t),
\end{eqnarray*}
where $a_1=-(1-p)\delta \gamma $, $b_1= -p\delta \gamma $, $c_1=-p\delta (1-2s^*)$, $d_1=-(1-p)\delta(1-2s^*)$, $e_1=p \delta $, $f_1=(1-p)\delta$.
The linearized equation is given by:
\begin{eqnarray*}
\frac{dx(t)}{dt}=a_1 x(t) +b_1 x(t-\tau).
\end{eqnarray*}
The associated characteristic equation is determined by:
\begin{eqnarray}\label{eq_145}
\lambda-b_1 exp(-\lambda \tau)-a_1=0.
\end{eqnarray}
From the characteristic equation above, we derive the following result on the local asymptotic stability of the equilibrium.
\begin{prop}\label{prop_di_1}
\begin{itemize}
\item[$\bullet$] If $p\leq 0.5$, then the mixed ESS is asymptotically stable in the replicator dynamics for any value of $\tau$,
\item[$\bullet$] If $p>0.5$, then a Hopf bifurcation exists, when $\tau=\tau_{cr}$, with $\tau_{cr}=\frac{\mbox{acos}(-\frac {1-p}{p})}{\delta\gamma \sqrt{2p-1}}$.
\end{itemize}
\end{prop}
\begin{proof}\label{prop0}
\begin{itemize}
\item[$\bullet$] Let $\lambda=u+iv$, where $v>0$ a root of (\ref{eq_145}). We suppose that $u>0$ and we aim to prove that $p>0.5$. Substituting $\lambda$ by $u+iv$ in equation (\ref{eq_145}) and separating the real and imaginary parts, we derive,
\begin{eqnarray*}
u +(1-p) \delta \gamma=-p\delta \gamma e^{-u\tau} \mbox{cos}(v\tau),\\
v=p\delta \gamma  e^{-u\tau} \mbox{sin}(v\tau).
\end{eqnarray*}
which yields,
\begin{eqnarray}
\big(u +(1-p) \delta \gamma \big)^2+v^2=p^2 \delta^2 \gamma^2 e^{-2u\tau}.
\end{eqnarray}
Since $u>0$, we conclude the following inequalities,
\begin{eqnarray*}
\big(u +(1-p) \delta \gamma \big)^2+v^2 \leq p^2 \delta^2 \gamma^2 , \\
\big((1-p) \delta \gamma \big)^2   \leq \big(u +(1-p) \delta \gamma \big)^2+v^2.
\end{eqnarray*}
Finally, from these inequalities, we obtain,
\begin{eqnarray*}
(1-p) < p, \mbox{ and } p > 0.5.
\end{eqnarray*}
Therefore, $u<0$ for any $p \leq 0.5$, and the asymptotic stability follows.
\item[$\bullet$]Let $\lambda^*=iw_0$ where $w_0>0$ be a root of the characteristic equation. From (\ref{eq_145}), we get,
\begin{eqnarray*}
iw_0+p\gamma \delta exp(-iw_0 \tau_{cr})+(1-p) \gamma \delta=0,
\end{eqnarray*}
which yields,
\begin{eqnarray*}
\mbox{cos}(w_0\tau_{cr})=-\frac{1-p}{p},\mbox{ and }
\mbox{sin}(w_0\tau_{cr})=\frac{w_0}{p\gamma \delta}.
\end{eqnarray*}
Finally, we get:
\begin{eqnarray*}
\tau_{cr}=\frac{\mbox{acos}(-\frac {1-p}{p})}{\delta\gamma \sqrt{2p-1}},\mbox{ and }
w_0=\delta\gamma \sqrt{2p-1},
\end{eqnarray*}
where $'\mbox{acos}'$ denotes the $0$ to $\pi$ branch of the inverse cosine function.
Furthermore, we have,
\begin{eqnarray*}
{\cal{R}}e \frac{d\lambda(\tau)}{d\tau}_{|\tau=\tau_{cr}}=\frac{w_0^2}{(1-a_1 \tau_{cr})^2+ \tau_{cr}^2 w_0^2}>0,
\end{eqnarray*}
which means that when $\tau$ is near $\tau_{cr}$ and $\tau>\tau_{cr}$, two roots gain positive parts as $\tau$ passes through $\tau_{cr}$.
Therefore, when $p\geq 0.5$, a Hopf bifurcation exists at $\tau_{cr}$.
\end{itemize}
\end{proof}
As a remark, we notice that when $p=1$, the critical delay is given by $\tau_{cr}=\frac{\pi}{2\delta\gamma}$ and this value coincides with that obtained in the Dirac distribution case.
\begin{figure}[t]
 \hskip -0.53cm\includegraphics[width=4.97cm]{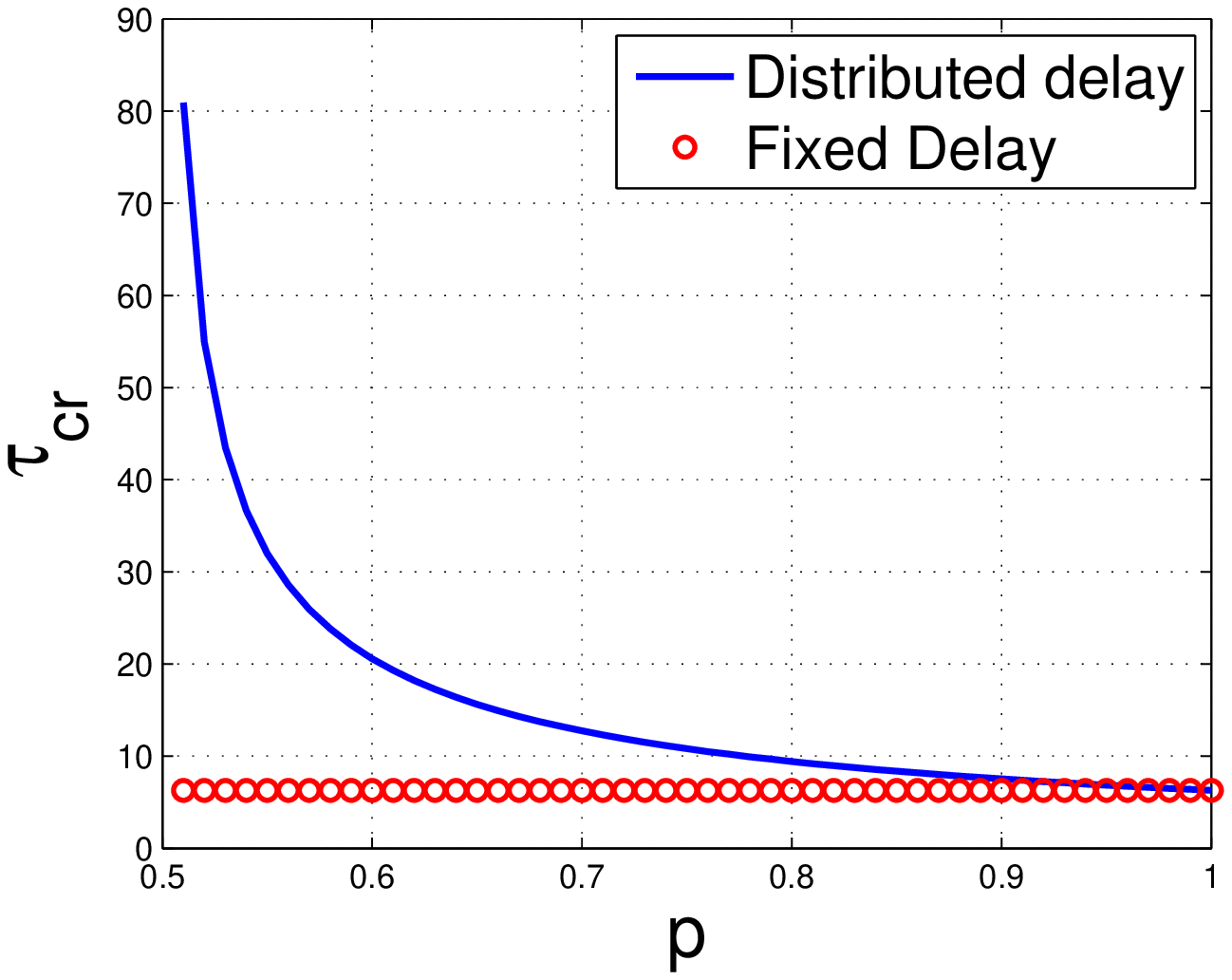}\includegraphics[width=4.97cm]{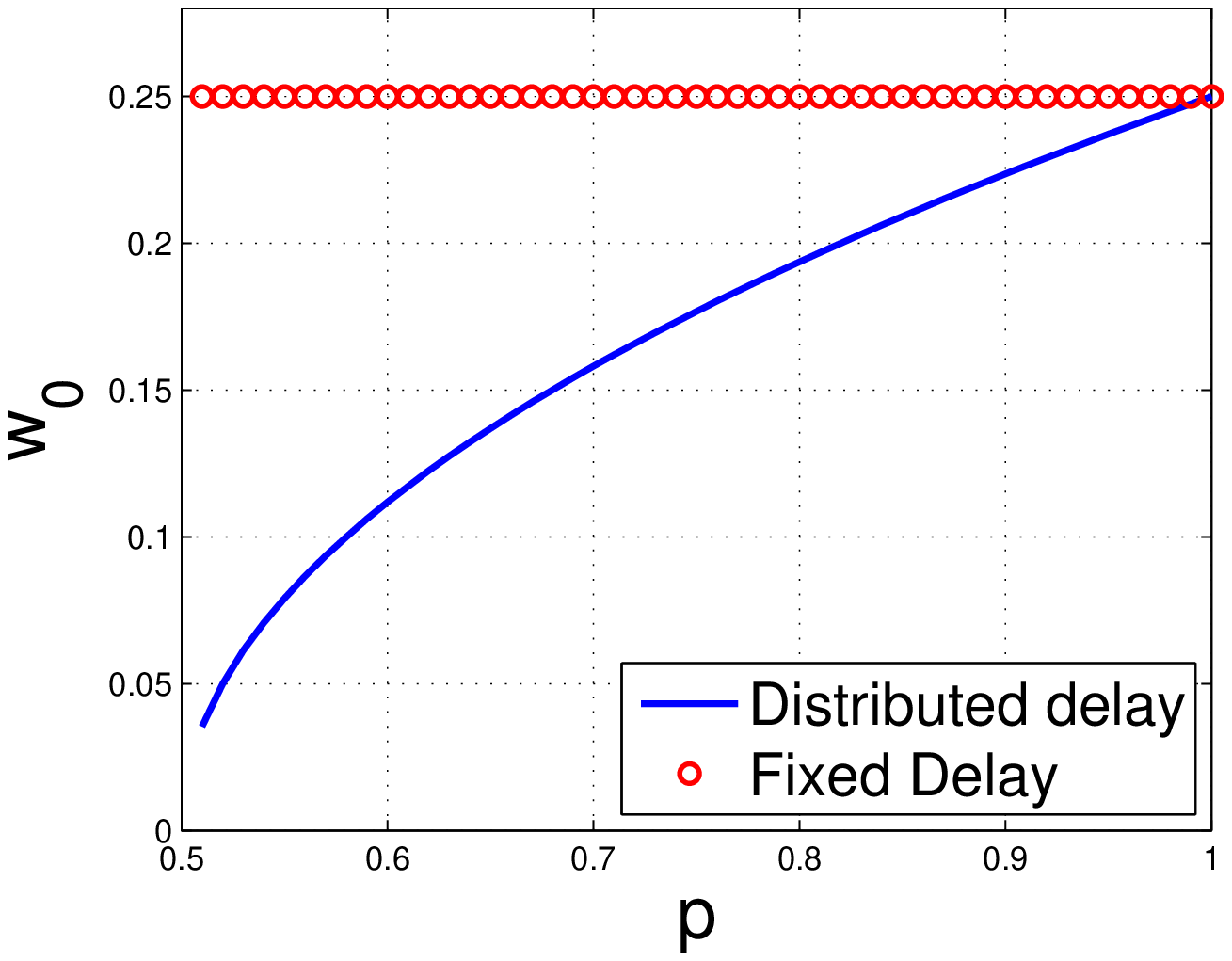} \caption{Left, the critical delay $\tau_{cr}$ in function of $p$. Right, the frequency of oscillations at the Hopf bifurcation, $w_0$, in function of $p$, where $a=-0.5$, $b=1$, $c=0$, and $d=0.5$.}\label{fig_discrete}
\end{figure}
In Fig. \ref{fig_discrete}, we plot the critical delay $\tau_{cr}$ and the frequency of oscillations at the Hopf bifurcation, in function of $p$, the probability of a delayed strategy. The range of $p$ at which a Hopf bifurcation may exist is $]0.5,1]$. We observe that as $p$ increases, the critical delay decreases, and thus the potential of instability increases. For instance, when $p=0.6$, the critical delay is given by $20.5$ time units, whereas this value decreases to $9.4$ time units when $p=0.8$. In addition, the frequency of oscillations at the Hopf bifurcation grows gradually as $p$ increases, which emphasizes the instability property. For example, when $p=0.6$, $w_0=0.11$, while $w_0=0.19$ when $p=0.8$. 

It is also interesting to compare the results in our scenario with those obtained in the classical case of a single deterministic delay. Therefore, we displayed in Fig. \ref{fig_discrete} the critical delay value (which we denote by $\tau_{c0}$ and the frequency of oscillation (which we denote by $w_{c0}$) in the case of a single delay.  We observe that $\tau_{cr}$ (as defined in proposition \ref{prop_di_1}) is always larger than $\tau_{c0}$ and they coincide only when $p=1$. Similarly, $w_0$ is always smaller than $w_{c0}$ and they coincide when $p=1$, in which case the two scenarios are exactly the same.

Furthermore, the properties of the bifurcating limit cycle are brought out in the next proposition.
\begin{prop}\label{prop_di_2}
Let $P$ and $Q$ be defined as follows:
\begin{eqnarray*}
P=4b_1^3(b_1-a_1)(a_1+b_1)^2(-5b_1+4a_1),
\end{eqnarray*}
and
\begin{eqnarray*}
Q&=&5e_1b_1^6 \tau_{cr}+a_1e_1b_1^5 \tau_{cr}-15a_1f_1b_1^5 \tau_{cr}-3c_1^2b_1^2 \tau_{cr}\\&&-7c_1d_1b_1^5\tau_{cr}-4d_1^2b_1^5 \tau_{cr}+
6 a_1^2e_1b_1^4\tau_{cr}-3a_1^2f_1b_1^4 \tau_{cr}\\&&+7c_1^2a_1b_1^4 \tau_{cr}+19c_1d_1a_1b_1^4 \tau_{cr}+18 d_1^2a_1b_1^4 \tau_{cr}+\\&&2a_1^3e_1b_1^3\tau_{cr}+12a_1^3f_1b_1^3\tau_{cr}-12c_1^2a_1^2b_1^3\tau_{cr}- \\&& 26c_1d_1a_1^2b_1^3\tau_{cr}-8d_1^2a_1^2b_1^3\tau_{cr}-8a_1^4e_1b_1^2\tau_{cr}+\\&&8c_1^2a_1^3b_1^2\tau_{cr}+
8c_1d_1a_1^3b_1^2\tau_{cr}+15f_1b_1^5-15a_1e_1b_1^4+\\&&3a_1f_1b_1^4-c_1^2b_1^4-9c_1d_1b_1^4-18d_1^2b_1^4-3a_1^2e_1b_1^3-\\&&12a_1^2f_1b_1^3+
11c_1^2a_1b_1^3+33c_1d_1a_1b_1^3+12d_1^2a_1b_1^3+\\&&12a_1^3e_1b_1^2-14c_1^2a_1^2b_1^2-18c_1d_1a_1^2b_1^2+4c_1^2a_1^3b_1.
\end{eqnarray*} 
The amplitude of the bifurcating limit cycle is given by:
\begin{eqnarray*}
A_m=\sqrt{\frac{P}{Q}\mu},
\end{eqnarray*}
where $\mu=\tau-\tau_{cr}$. Furthermore, the Hopf bifurcation is supercritical.
\end{prop}
\begin{proof}
The proof follows by carrying out the same procedure as in the previous sections.
\end{proof}
This proposition gives a closed-form expression of the amplitude of the bifurcating periodic solution. Indeed, when $\tau<\tau_{cr}$, the mixed equilibrium $s^*$ is asymptotically stable, whereas, for the values of $\tau$ near and superior to $\tau_{cr}$, a stable periodic oscillation appears with an amplitude proportional to $\sqrt{\tau-\tau_{cr}}$.
\section{Numerical simulations}\label{sec_num}
\begin{figure}[t]
\hskip -0.54cm \includegraphics[width=4.8cm]{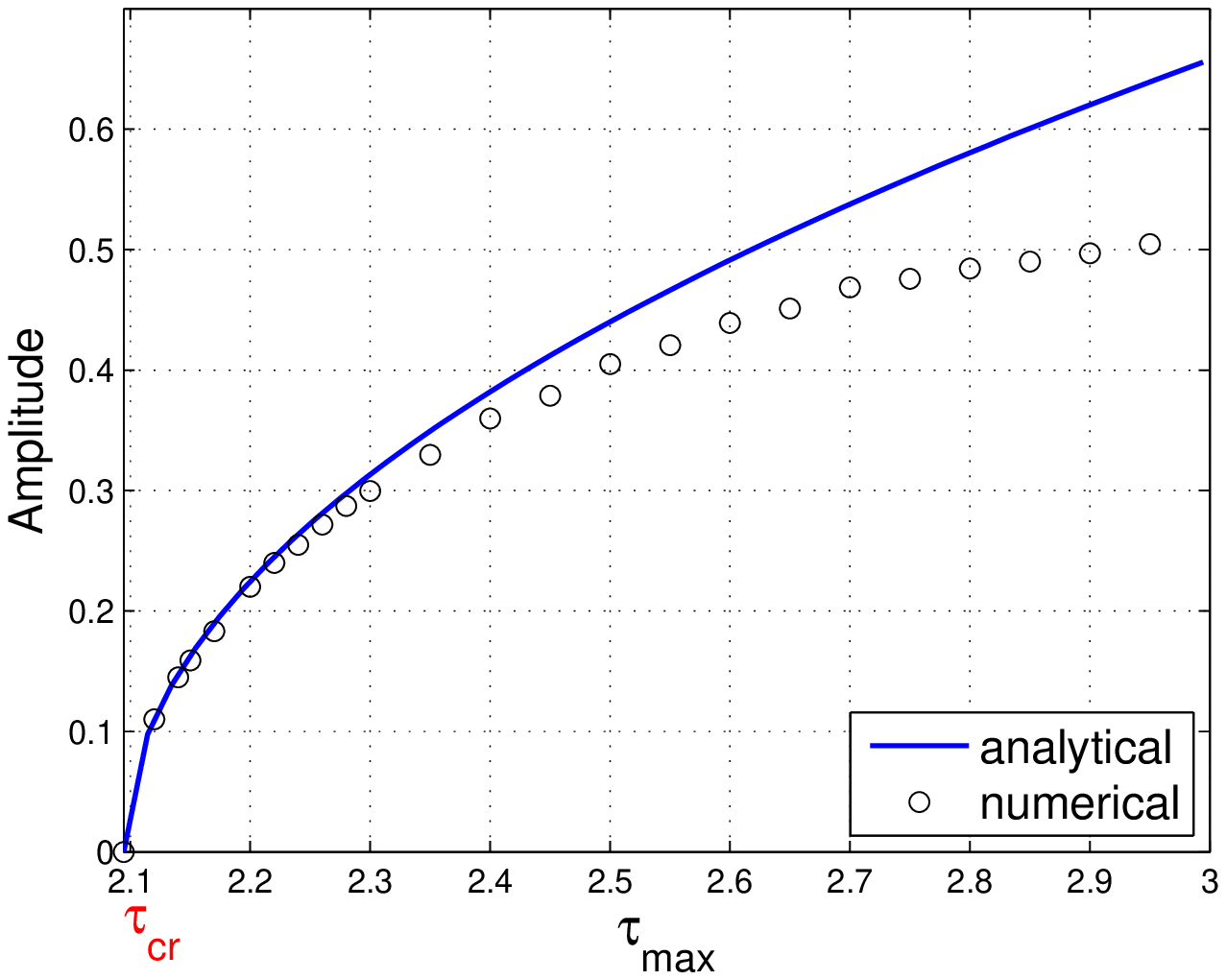}\includegraphics[width=4.8cm]{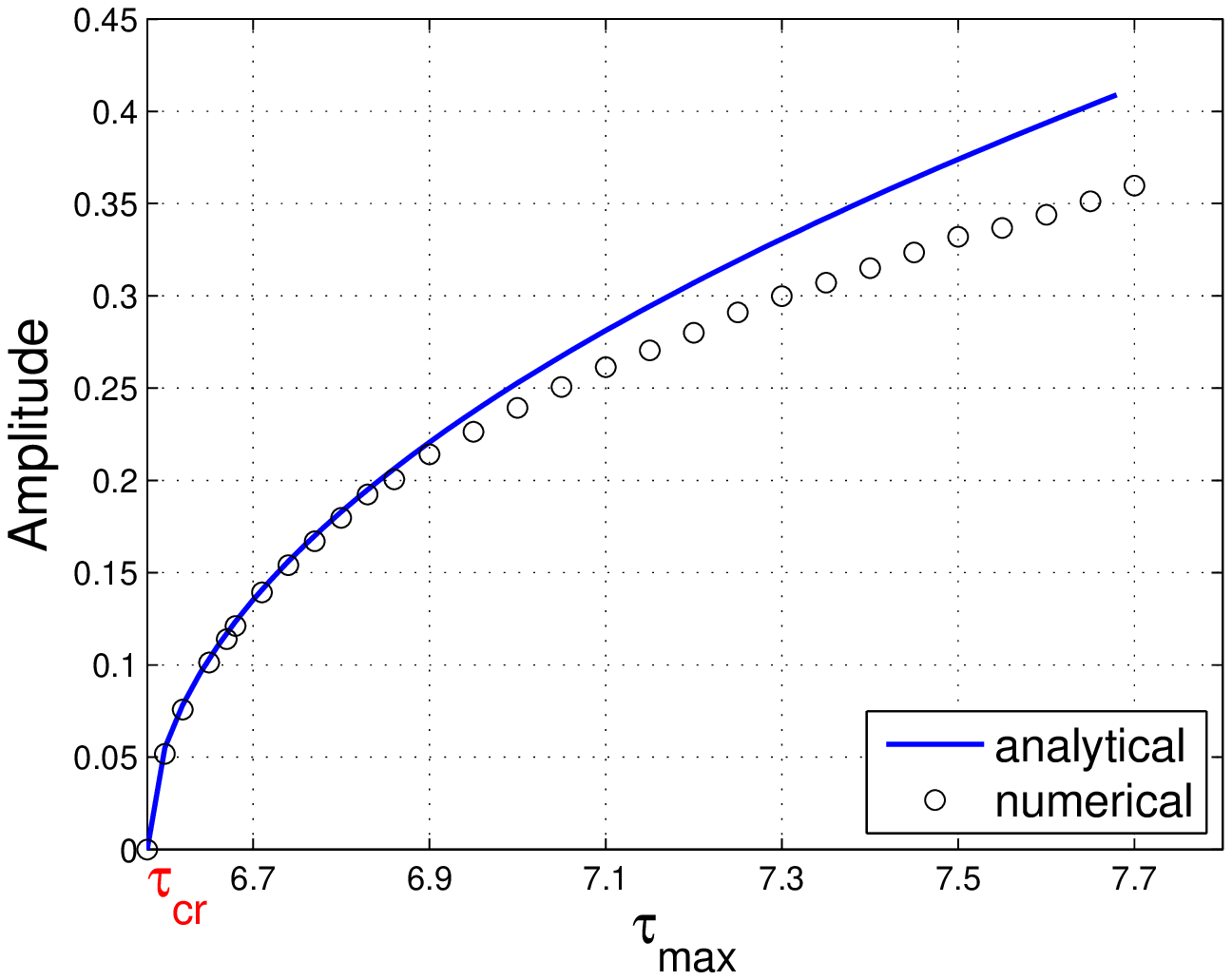}
\mbox{\hskip -0.54cm}  \includegraphics[width=4.8cm]{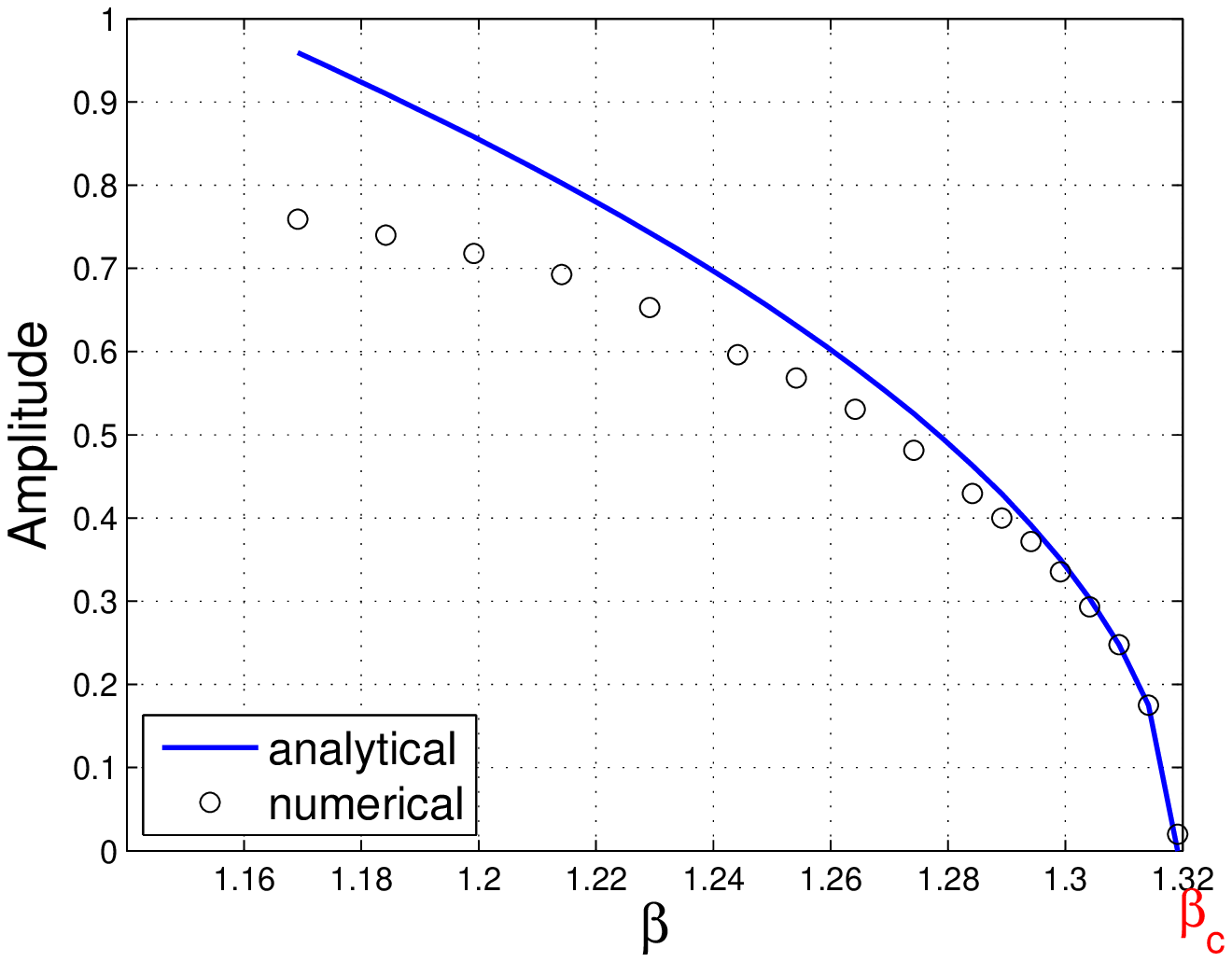}\includegraphics[width=4.8cm]{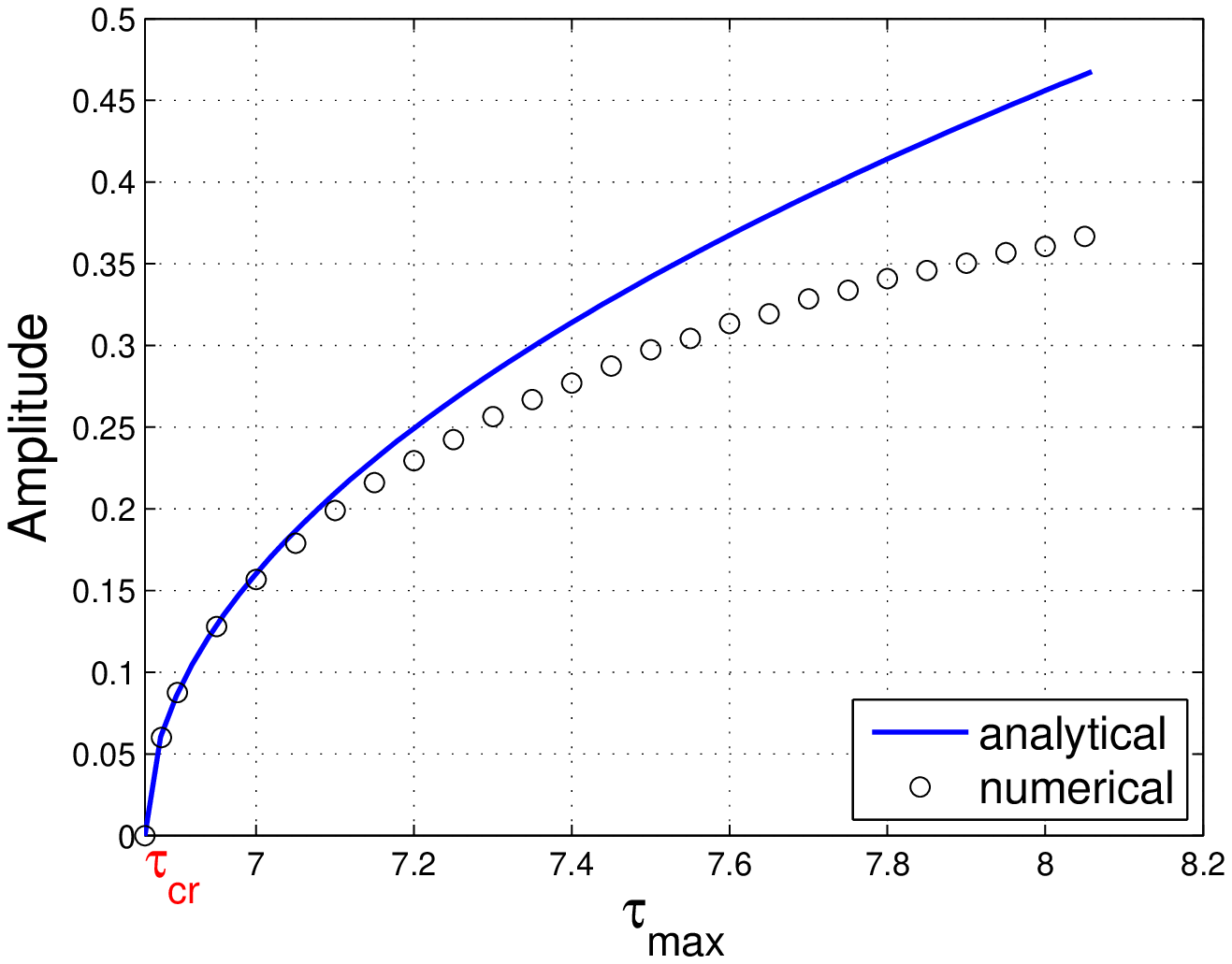}
\caption{The amplitude of the bifurcating periodic solution near the Hopf bifurcation, where $a=-1.5$, $b=3$, $c=0$, and $d=1.5$. {\it{Top-left}}, Dirac distribution. {\it{Top-right}}, Uniform distribution. {\it{Bottom-left}}, Gamma distribution with $k=3$. {\it{Bottom-right}}, Discrete distribution with $p=0.6$.}\label{fig_amplitude1}
\end{figure}
In this section, we propose to compare the properties of the bifurcating periodic solution obtained by  the perturbation method with numerical results. We depict in Fig. \ref{fig_amplitude1} the amplitude of the bifurcating limit cycle given in propositions \ref{prop_dirac}-\ref{prop_di_2} and the amplitude obtained numerically (in circles), for different delay distributions. We observe that the results of the two methods coincide for the values of $\tau$ (or $\beta$) close to $\tau_{cr}$ (or $\beta_c$). For example, 
in the case of discrete delays, the critical delay is given by $\tau_{cr}=6.86$ time units, and for the values of delays close to $\tau_{cr}$ the amplitude predicted analytically and the one obtained numerically coincide but the difference between them increases gradually until reaching $0.1$ when $\tau_{max}$  equals $8.05$ time units. For the Gamma distribution, $\beta_c=1.32$ and the critical mean delay is given by $2.28$ time units. The bifurcation occurs for the values of $\beta$ near and below $\beta_c$ (recall that the mean is $\frac{k}{\beta}$), which explains the shape of the amplitude growth.

\section{CONCLUSIONS}\label{sec_conc}
In this paper, we considered the two-strategy replicator dynamics subject to uncertain delays. Taking as a bifurcation parameter the mean delay, we proved that the asymptotic stability of the mixed equilibrium may be lost at the Hopf bifurcation, in which case the replicator dynamics exhibits a stable periodic oscillation (limit cycle) in the proportion of strategies in the population. As the mean delay moves away from the critical value, the amplitude of the limit cycle grows gradually. Using a nonlinear Lindstedt's perturbation method and considering different probability distributions of delays, we approximated the bifurcating limit cycle and we determined analytically the growth rate of the radius of the limit cycle. Furthermore, we compared with numerical simulations. As an extension to this work, we plan to investigate the center manifold approach.

\addtolength{\textheight}{-12cm}   



\bibliographystyle{IEEEtran}
\bibliography{ref} 
\section*{APPENDIX}
\subsection{Proof of Proposition \ref{prop_dirac}}
Since we take as a bifurcation parameter $\tau$, we make the following series expansion:
\begin{eqnarray*}
\tau&=&\tau_{cr}+\epsilon^2 \hat{\mu} + {{O}}(\epsilon^3),\nonumber\\
&=&\tau_{cr}+\mu +{{O}}(\epsilon^3).
\end{eqnarray*}
Furthermore, we expand $u(T-\Omega \tau)$ as follows :
\begin{eqnarray*}
u(T-\Omega \tau)&=&u(T-(w_0 + \epsilon^2 k_2+..)(\tau_{cr}+\epsilon^2 \hat{\mu})) \nonumber \\
&=&u(T-w_0 \tau_{cr}- \epsilon^2(k_2 \tau_{cr}+w_0 \hat{\mu}) +...)\nonumber \\
&=&u(T-w_0 \tau_{cr})-\epsilon^2(k_2 \tau_{cr}+w_0 \hat{\mu}) \times \nonumber \\ &&u'(T-w_0\tau_{cr})+{{O}}(\epsilon^3).
\end{eqnarray*}
In addition, we have:
\begin{eqnarray*}
\Omega=w_0+ \epsilon^2k_2+{{O}}(\epsilon^3).
\end{eqnarray*}
Substituting the above series expansions and collecting terms of similar order in $\epsilon$ in equation (\ref{repl_dirac1}), we get:
\begin{eqnarray}
&& \hskip -0.912cm \bullet \;\,  w_0 \frac{du_0(T)}{dT}-A u_0(T-w_0 \tau_{cr})=0,\label{eqrrr}\\
&& \hskip -0.912cm \bullet \;\,  w_0 \frac{du_1(T)}{dT}-A u_1(T-w_0\tau_{cr})=B u_0(T)u_0(T-w_0\tau_{cr}),\;\;\;\;\; \label{equ_11}\\
&& \hskip -0.912cm \bullet \;\,  w_0 \frac{du_2(T)}{dT}-A  u_2(T-w_0\tau_{cr})=-k_2 \frac{du_0(T)}{dT}-A(k_2 \tau_{cr}+ \nonumber \\ && \hskip -0.4cm w_0 \hat{\mu}) u_0'(T-w_0\tau_{cr})+B u_1(T) u_0(T-w_0\tau_{cr})+\nonumber \\ &&\hskip -0.4cm Bu_0(T) u_1(T-w_0\tau_{cr})+C u_0(T-w_0\tau_{cr}) u_0^2(T) \label{eq_sec1}.
\end{eqnarray}
A solution of (\ref{eqrrr}) is of the form,
\begin{eqnarray}\label{equ_0}
u_0(T)={\hat{A}}_m {\mbox{cos}}(T).
\end{eqnarray}
 We substitute (\ref{equ_0}) into  (\ref{equ_11}) to get the following equation in $u_1$:
\begin{eqnarray*}
w_0 \frac{du_1(T)}{dT}-A u_1(T-w_0\tau_{cr})=\frac{B {\hat{A}}_m^2}{2} {\mbox{sin}}(2T).
\end{eqnarray*}
Let $u_1(T)=m_1 {\mbox{sin}}(2T)+m_2 {\mbox{cos}}(2T)$ be a solution of the above DDE. Solving this equation for $u_1$ yields:
\begin{eqnarray*}
m_1={\hat{A}}_m^2 \frac{B}{10 A}, \quad \mbox{and} \quad  m_2=2m_1.
\end{eqnarray*}
We used the relation $w_0=\delta \gamma=-A$. 
Finally, after using some trigonometric formulae and setting the secular terms ($\mbox{cos}(T)$ and $\mbox{sin}(T)$ terms that yield a resonance effect) to zero in equation (\ref{eq_sec1}), we get the amplitude of the bifurcating periodic solution:
\begin{eqnarray*}
{\hat{A}}_m^2=\frac{-20A^3 \hat{\mu}}{5CA^2\tau_{cr}-3B^2A \tau_{cr}-B^2}.
\end{eqnarray*}
Multiplying both sides by $\epsilon^2$ in the equation above, we get:
\begin{eqnarray*}
{{A}}_m^2=\frac{-20A^3 {\mu}}{5CA^2\tau_{cr}-3B^2A \tau_{cr}-B^2}.
\end{eqnarray*}
Since the equilibrium is asymptotically stable when $\tau<\tau_{cr}$ and $\frac{d\lambda (\tau)}{d\tau}_{|\tau=\tau_{cr}}>0$, the limit cycle is stable and the bifurcation is supercritical.
\subsection{Proof of Proposition \ref{prop_unif}}
\begin{proof}
Since we take as a bifurcation parameter $\tau_{max}$, we can make the following series expansion:
\begin{eqnarray*}
\tau_{max}&=&\tau_{cr}+\epsilon^2 \hat{\mu} +{{O}}(\epsilon^3), \nonumber \\
&=&\tau_{cr}+\mu +{{O}}(\epsilon^3).
\end{eqnarray*}
We recall that,
\begin{eqnarray*}
u(T-\Omega \tau)&=&u(T- (w_0 +k_2 \epsilon^2+..) \tau),\nonumber \\
&=&u(T-w_0 \tau)-k_2 \epsilon ^2 \tau u'(T-w_0 \tau) +{{O}}(\epsilon^3).
\end{eqnarray*}
Let $\bar{u}=u(T-\Omega \tau)$. We have,
\begin{eqnarray*}
&&\hskip -0.9cm \int_{0}^{\tau_{max}}{\bar{u}d\tau}=\int_{0}^{\tau_{cr}}{\bar{u}d\tau}+\int_{\tau_{cr}}^{\tau_{cr}+\epsilon^2 \hat{\mu} }{\bar{u}d\tau},\nonumber\\
&&=\int_{0}^{\tau_{cr}}{u_0(T-w_0 \tau)d\tau}+\epsilon \int_{0}^{\tau_{cr}}{u_1(T-w_0 \tau)d\tau}+\nonumber \\ && \epsilon^2 \Big(\int_{0}^{\tau_{cr}}{\big(u_2(T-w_0 \tau)-k_2 \tau u_0'(T-w_0 \tau) \big)d\tau}\nonumber \\ && +\hat{\mu} u_0(T-w_0 \tau)\Big)+{{O}}(\epsilon^3).
\end{eqnarray*}
Taking account of the above series expansions, we collect terms of the same order in $\epsilon$ in (\ref{repl_unif}), we get the following system of equations, which we resolve recursively:
\begin{eqnarray}
&& \hskip -0.8cm  \bullet \;\; w_0 \tau_{cr} \frac{du_0(T)}{dT}-A \int_{0}^{\tau_{cr}}{u_0(T-w_0\tau) d\tau}=0, \label{eq_uuuu}\\
&& \hskip -0.8cm  \bullet \;\; w_0 \tau_{cr} \frac{du_1(T)}{dT}-A \int_{0}^{\tau_{cr}}{u_1(T-w_0\tau)d\tau}=B u_0(T)  \times \nonumber \\ && \int_{0}^{\tau_{cr}}{u_0(T-w_0\tau)d\tau}, \label{eq_99} \\
&& \hskip -0.8cm  \bullet \;\;w_0\tau_{cr} \frac{du_2(T)}{dT}-A\int_{0}^{\tau_{cr}}{u_2(T-w_0\tau)d\tau}=-(k_2 \tau_{cr}+\nonumber \\ && \hskip -0.2cm w_0 \hat{\mu})\frac{du_0(T)}{dT} + (A\hat{\mu}+ \frac{Ak_2 \tau_{cr}}{w_0}) u_0(T-w_0 \tau_{cr})+ \nonumber \\&& \hskip -0.7 cm (-\frac{A k_2}{w_0} +Bu_1(T)+ C u_0^2(T) ) \int_{0}^{\tau_{cr}}{u_0(T-w_0\tau)d\tau} \label{eq_1012}.
\end{eqnarray}
A solution of (\ref{eq_uuuu}) is of the form:
\begin{eqnarray*}
u_0(T)=\hat{A}_m \mbox{cos}(T).
\end{eqnarray*}
Let $u_1(T)=m_1{\mbox{sin}}(2T)+m_2{\mbox{cos}}(2T)$ be a solution of (\ref{eq_99}). Then, we have, 
\begin{eqnarray*}
\int_{0}^{\tau_{cr}}{u_1(T-w_0\tau)d\tau}=0.
\end{eqnarray*}
The result above can be obtained by remembering that $\mbox{cos}(w_0\tau_{cr})=-1$ and $\mbox{sin}(w_0\tau_{cr})=0$.
On the other hand, we have,
\begin{eqnarray*}
\int_{0}^{\tau_{cr}}{u_0(T-w_0\tau)d\tau}&=&2\frac{\hat{A}_m}{w_0}{\mbox{sin}}(T).
\end{eqnarray*}
Therefore, equation (\ref{eq_99}) can be reduced to:
\begin{eqnarray*}
w_0 \tau_{cr} \frac{du_1(T)}{dT}=\frac{\hat{A}_m^2 B}{w_0} {\mbox{sin}}(2T)
\end{eqnarray*}
Considering that $u_1(T)=m_1{\mbox{sin}}(2T)+m_2{\mbox{cos}}(2T)$, we get:
\begin{eqnarray*}
m_1=0, \quad \mbox{and} \quad  m_2=-\frac{B \hat{A}_m^2}{2w_0^2\tau_{cr}}.
\end{eqnarray*}
On the other hand, equation (\ref{eq_1012}) can be written as:
\begin{eqnarray*}
w_0\tau_{cr} \frac{du_2(T)}{dT}-A \int_{0}^{\tau_{cr}}{u_2(T-w_0\tau)d\tau}=\hat{A}_m{{F}}_1{\mbox{sin}}(T) \nonumber \\ + \hat{A}_m {{F}}_2 {\mbox{cos}}(T)+ {{F}}_3 \mbox{sin}(3T)+{{F}}_4 \mbox{cos}(3T).
\end{eqnarray*}
where,
\begin{eqnarray*}
{{F}}_1= (k_2\tau_{cr}+\hat{\mu} w_0)-2\frac{Ak_2}{w_0^2}-\frac{B m_2}{w_0}+\frac{C \hat{A}_m^2}{2w_0},
\end{eqnarray*}
and,
\begin{eqnarray*}
{{F}}_2=-A\hat{\mu}-A \frac{k_2\tau_{cr}}{w_0}.
\end{eqnarray*}
By removing the secular terms that yield a resonant effect, that is by setting ${{F}}_1=0$ and ${{F}}_2=0$, we get:
\begin{eqnarray*}
\hat{A}_m^2=\frac{8A^2\hat{\mu}}{\tau_{cr}(B^2-2CA)}.
\end{eqnarray*}
Multiplying both sides by $\epsilon^2$ in the equation above, we get:
\begin{eqnarray*}
{A}_m^2=\frac{8A^2{\mu}}{\tau_{cr}(B^2-2CA)}.
\end{eqnarray*}
Since the equilibrium is asymptotically stable when $\tau_{max}<\tau_{cr}$ and $\frac{d\lambda(\tau_{max})}{d\tau_{max}}_{|\tau_{max}=\tau_{cr}}>0$, the limit cycle is stable and the Hopf bifurcation is supercritical.
\end{proof}
\subsection*{Proof of Proposition \ref{prop_gamma}}
Since we take as a bifurcation parameter $\beta$, we can make a series expansion of $\beta$ as follows:
\begin{eqnarray*}
\beta=\beta_c + \epsilon^2 \hat{\mu}+ {{O}}(\epsilon^3).
\end{eqnarray*}
By remembering that $(\beta_c +\epsilon^2 \hat{\mu})^k=\beta_c^k+k\beta_c^{k-1}\epsilon^2 \hat{\mu}+{{O}}(\epsilon^3)$, and by making a series expansion in (\ref{eq_0001}), we get,
\begin{eqnarray*}
{{I}}&=&\int_{0}^{\infty}{\frac{{(\beta_c +\epsilon^2 \mu)}^k}{\Gamma(k)} \tau^{k-1}e^{-(\beta_c+\epsilon^2 \mu) \tau}u(T-\Omega \tau) d\tau}\nonumber \\
&=&\int_{0}^{\infty}{\frac{{(\beta_c +\epsilon^2 \mu)}^k}{\Gamma(k)} \tau^{k-1}e^{-\beta_c \tau}(1-\epsilon^2 \mu \tau)u(T-\Omega \tau) d\tau}\nonumber \\
&=&\frac{\beta_c^{k}}{\Gamma(k)}\int_{0}^{\infty}{\tau^{k-1}e^{-\beta_c \tau}u_0(T-w_0 \tau)d\tau}+\epsilon\frac{\beta_c^{k}}{\Gamma(k)}\times \\ &&\int_{0}^{\infty}{\tau^{k-1}e^{-\beta_c \tau}  u_1(T-w_0 \tau) d\tau}+\frac{\epsilon^2 \beta_c^{k}}{\Gamma(k)} \times \nonumber \\ &&\hskip -0.9cm {\scriptsize {\int_{0}^{\infty}{\tau^{k-1}e^{-\beta_c \tau}\Big(u_2(T-w_0\tau)-k_2\tau u_0'(T-w_0\tau)}}}-\\&&{{{\hat{\mu}  \tau u_0(T-w_0\tau)\Big)d\tau}}}+\frac{\epsilon^2k\beta_c^{k-1} \hat{\mu}}{\Gamma(k)} \times \nonumber \\ &&  \int_{0}^{\infty}{ \tau^{k-1}e^{-\beta_c \tau} u_0(T-w_0 \tau) d\tau}+{{O}}(\epsilon^3).
\end{eqnarray*}
Then, taking account of the previous expansions, and collecting terms of similar order in $\epsilon$, we get the following equations which we resolve recursively:
\begin{eqnarray}
&&\hskip -0.8 cm \bullet \;\; w_0 \frac{du_0(T)}{dT}-A \frac{\beta_c^k}{\Gamma(k)} \int_{0}^{\infty}{\tau^{k-1} e^{-\beta_c \tau}u_0(T-w_0\tau)d\tau}=0,\nonumber\\
&&\hskip -0.8 cm \bullet \;\; w_0 \frac{du_1(T)}{dT}-A \frac{\beta_c^k}{\Gamma(k)} \int_{0}^{\infty}{\tau^{k-1} e^{-\beta_c \tau}u_1(T-w_0\tau)d\tau}=\nonumber \\ && Bu_0(T) \frac{\beta_c^k}{\Gamma(k)} \int_{0}^{\infty}{\tau^{k-1} e^{-\beta_c \tau}u_0(T-w_0\tau)d\tau}, \nonumber \\ \label{eq_rr_t}\\
&&\hskip -0.8 cm\bullet \; \;  w_0 \frac{du_2(T)}{dT}-A\frac{\beta_c^k}{\Gamma(k)} \int_{0}^{\infty}{\tau^{k-1} e^{-\beta_c \tau}u_2(T-w_0\tau)d\tau}=\nonumber \\ &&-k_2\frac{du_0(T)}{dT}-Ak_2\frac{\beta_c^k}{\Gamma(k)}\int_{0}^{\infty}{\tau^k e^{-\beta_c\tau} u_0'(T-w_0\tau)d\tau}\nonumber \\
&&-\frac{A\beta_c^k \mu}{\Gamma(k)}\int_{0}^{\infty}{\tau^k e^{-\beta_c\tau} u_0(T-w_0\tau)d\tau}+\frac{B \beta_c^k}{\Gamma(k)} u_0(T) \times \nonumber \\ && \hskip -0.7cm \int_{0}^{\infty}{\tau^{k-1} e^{-\beta_c \tau}u_1(T-w_0\tau)d\tau}+ \Big( \frac{B \beta_c^k}{\Gamma(k)} u_1(T) + \frac{C\beta_c^k}{\Gamma(k)}u_0^2(t)\nonumber \\ && +\frac{kA}{\Gamma(k)} \beta_c^{k-1} \hat{\mu} \Big )\int_{0}^{\infty}{\tau^{k-1} e^{-\beta_c \tau}u_0(T-w_0\tau)d\tau}.  \label{eq_f}
\end{eqnarray}
Let $u_1=m_1\mbox{sin}(2T)+m_2\mbox{cos}(2T)$ be a solution of (\ref{eq_rr_t}). 
Solving (\ref{eq_rr_t}) in $u_1$, we get:
\begin{eqnarray*}
m_1= {{F}}_1  \hat{A}_m^2, \; \mbox{and} \quad m_2={{F}}_2 \hat{A}_m^2,
\end{eqnarray*}
where,
\begin{eqnarray*}\label{f_1}
\hskip -0.3cm{{F}}_1=-\frac{\frac{AB}{2} (1+\frac{w_0^2}{\beta_c^2})^{-\frac{k}{2}} (1+4\frac{w_0^2}{\beta_c^2})^{-\frac{k}{2}} \mbox{cos}(k\theta_1)}{4w_0^2+A^2(1+4\frac{w_0^2}{\beta_c^2})^{-k}+4w_0A (1+4\frac{w_0^2}{\beta_c^2})^{-\frac{k}{2}}\mbox{sin}(k\theta_1)},
\end{eqnarray*}
and,
\begin{eqnarray*}\label{f_2}
\hskip -0.3cm {{F}}_2=-\frac{\frac{B}{2} (1+\frac{w_0^2}{\beta_c^2})^{-\frac{k}{2}}(2w_0+A (1+4\frac{w_0^2}{\beta_c^2})^{-\frac{k}{2}} \mbox{sin}(k\theta_1)) }{4w_0^2+A^2(1+4\frac{w_0^2}{\beta_c^2})^{-k}+4w_0A (1+4\frac{w_0^2}{\beta_c^2})^{-\frac{k}{2}}\mbox{sin}(k\theta_1)},
\end{eqnarray*}
$\theta_1=\mbox{atan}(\frac{2w_0}{\beta_c})$, and '\mbox{atan}' denotes the  $0$ to $\frac{\pi}{2}$ branch of the inverse tangente function.
On the other hand, equation (\ref{eq_f}) can be written as:
\begin{eqnarray*}
w_0 \frac{du_2(T)}{dT}-A\frac{\beta_c^k}{\Gamma(k)} \int_{0}^{\infty}{\tau^{k-1} e^{-\beta_c \tau}u_2(T-w_0\tau)d\tau}= \\G\mbox{sin}(T)+K\mbox{cos}(T)+L \mbox{cos}(3T)+M\mbox{sin}(3T).
\end{eqnarray*}
where,
\begin{eqnarray*}
{{G}}=k_2\hat{A}_m-A\hat{A}_m\beta_c^k(k+1)(\beta_c^2+w_0^2)^{-\frac{k+2}{2}} (k_2w_0+\beta_c\hat{\mu}) \\ +\frac{C\beta_c^k \hat{A}_m^3}{4}(\beta_c^2+w_0^2)^{-\frac{k}{2}} 
+\hat{A}_m\beta_c^{k-1} (\beta_c^2+w_0^2)^{-\frac{k+1}{2}} (-\frac{m_2}{2} \times \\ B \beta_c w_0+ \frac{m_1}{2}B \beta_c^2+kA\hat{\mu}w_0),
\end{eqnarray*}
and
\begin{eqnarray*}
{{K}}=\beta_c^k (\beta_c^2+w_0^2)^{-\frac{k+1}{2}} \hat{A}_m (-k_2\beta_cA(k+1)(\beta_c^2+w_0^2)^{-\frac{1}{2}} +\\Aw_0 \hat{\mu} (k+1) (\beta_c^2+w_0^2)^{-\frac{1}{2}}+\frac{m_2}{2}B\beta_c+\frac{m_1}{2}w_0B+kA\hat{\mu}).
\end{eqnarray*}
Finally, by setting $K=0$ and $G=0$, we obtain the amplitude given in equation (\ref{amp_gamma}). Since the equilibrium is asymptotically stable when $\beta>\beta_c$ and $\frac{d\lambda(\beta)}{d\beta}_{|\beta=\beta_c}<0$, the bifurcating limit cycle is stable and the bifurcation is supercritical.

\end{document}